\documentclass[11pt]{article}

\usepackage{amsmath}
\usepackage{amssymb}
\usepackage{amsfonts}
\usepackage{amsthm} 
\usepackage{color}
\usepackage{url}
\usepackage{comment}
\usepackage{caption}
\usepackage{cite}
\usepackage{graphicx}
\def\phi{\varphi}

\newtheorem{theorem}{Theorem}
\newtheorem{lemma}[theorem]{Lemma}

\newtheorem{corollary}[theorem]{Corollary}

\newtheorem{definition}[theorem]{Definition}

\renewcommand{\a}{\alpha}
\renewcommand{\b}{\beta}
\newcommand{\la}{<_\alpha}
\newcommand{\lb}{<_\beta}
\newcommand{\ga}{>_\alpha}
\newcommand{\gb}{>_\beta}

\title{ Existence of EFX for Two Additive Valuations }
\author{Ryoga Mahara\thanks{Research Institute for Mathematical Sciences, Kyoto University, Kyoto, 606-8502, Japan.
E-mail: ryoga@kurims.kyoto-u.ac.jp}
}
\date{}
\begin{document}
\maketitle
\begin{abstract}
Fair division of indivisible items is a well-studied topic in Economics and
Computer Science.
The objective is to allocate items to agents in a {\it fair} manner, 
where each agent has a valuation for each subset of items.
Envy-freeness is one of the most widely studied notions of fairness.
Since complete envy-free allocations do not always exist when items are indivisible,
several relaxations have been considered.
Among them, possibly the most compelling one is envy-freeness up to any item (EFX), where no agent envies another agent after the removal of any single item from the other agent's bundle.
However, despite significant efforts by many researchers for several years, 
it is known that a complete EFX allocation always exists only in limited cases.
In this paper, we show that a complete EFX allocation always exists when each agent is of one of two given types, where agents of the same type have identical additive valuations.
This is the first such existence result for non-identical valuations when there are any number of agents and items and no limit on the number of distinct values an agent can have for individual items.
We give a constructive proof, in which we iteratively obtain a Pareto dominating (partial) EFX allocation from an existing partial EFX allocation.
\end{abstract}


\section{Introduction} 
\label{sec: intro}
Fair division of items among competing agents is a fundamental and well-studied problem in Economics and
Computer Science. 
We are given a set $M$ of $m$ items and a set $N$ of $n$ agents with individual preferences.
The goal is to allocate items among $n$ agents in a {\it fair} manner.
When allocating all items to agents, such an allocation is called {\it complete}.
Otherwise, an allocation is called {\it partial}.
In this paper, we consider the {\it indivisible} setting: 
an item can not be split among multiple agents.
Several concepts of fairness have been considered in the literature, 
and one of the most well-studied notions of fairness is {\it envy-freeness}.
Each agent $i$ has a valuation function $v_i: 2^{M} \rightarrow \mathbb{R}_{\ge 0}$ for each subset of items.
An allocation is {\it envy-free} if the utility of each agent's bundle is at least as much as that of any other agent.
Unfortunately, complete envy-free allocations do not always exist when items are indivisible.
We can easily see this even with two players and a single item having positive utility for both of them: 
one of the agents has to receive the item and the other agent envies her.
This motivates the study of relaxations of envy-freeness.
\subsection{Previous Work}
\paragraph{Envy-freeness up to one item (EF1):}
Budish~\cite{budish2011combinatorial} introduced the notion of EF1, where no agent $i$ envies another agent $j$ after the removal of {\it some} item in $j$'s bundle.
That is, in an EF1 allocation, agent $i$ may envy agent $j$, but removing some item from $j$'s bundle would guarantee that $i$ does not envy $j$.
Note that no item is really removed from the envied agent's bundle, 
this is just a thought experiment to quantify the envy that the envious agent has toward the envied agent.
It is shown by Lipton et~al.~\cite{lipton2004approximately} that a complete EF1 allocation always exists and it can be obtained in polynomial time.

\paragraph{Envy-freeness up to any item (EFX):}
Plaut and Roughgarden~\cite{plaut2020almost} introduced the notion of EFX, where no agent $i$ envies another agent $j$ after the removal of {\it any} item in $j$'s bundle.
EFX is strictly stronger than EF1 but strictly weaker than full envy-freeness.
Caragiannis et al.~\cite{caragiannis2019envy} remarked that {\it ``Arguably, EFX is the best fairness analog of envy-freeness for indivisible items.''}
However, while it is known that EF1 allocations always exist, 
the existence of complete EFX allocations is not known except in limited cases.
As described in ~\cite{caragiannis2019unreasonable}, 
{\it ``Despite significant effort, we were not able to settle the question of whether an EFX allocation always exists (assuming all goods must be allocated), and leave it as an enigmatic open question.''}

Plaut and Roughgarden~\cite{plaut2020almost} showed that complete EFX allocation always exists (i) when there are only two agents or (ii) when all agents have identical valuations.
Furthermore, it was shown in~\cite{plaut2020almost} that exponentially many value queries may be required to identify EFX allocations even in the case where there are only two agents with identical {\it submodular} valuation functions.

Recently, Chaudhury et al.~\cite{chaudhury2020efx} showed that an EFX allocation always exists for three agents with {\it additive} valuations, that is, $v_i(S) = \sum_{g \in S} v_i(\{g\})$ for any agent $i$ and any subset $S \subseteq M$.
It is not known whether a complete EFX allocation always exists when there are more than three agents with additive valuations. 

As other recent results, Amanatidis et al.~\cite{amanatidis2020maximum} showed that a complete EFX allocation always exists when their valuations are additive and there are at most two possible values for the items.

\subsection{Our Contribution}
Our contribution in this paper is to show that a complete EFX allocation always exists when each agent is of one of two given types, where agents of the same type have identical additive valuation functions.
That is, we show the existence of EFX when the set of all agents $N$ is divided into $N_{\a}$ and $N_{\b}$, and each agent who belongs to $N_{\a}$ has additive valuation $v_{\a}$ and each agent who belongs to $N_{\b}$ has additive valuation $v_{\b}$. 
\begin{theorem}\label{thm: main}
Complete EFX allocations always exist when each agent is of one of two given types, where agents of the same type have identical additive valuation functions.
\end{theorem}
As mentioned above, Plaut and Roughgarden~\cite{plaut2020almost} showed that complete EFX allocations always exist when all agents have identical valuations.
Under the assumption that each valuation is additive, Theorem~\ref{thm: main} extends this result in the sense that the number of valuation types becomes two.

\subsection{Related Work}
Recently Caragiannis et al.~\cite{caragiannis2019envy} introduced a relaxation of EFX called {\it EFX with charity}.
This is a partial EFX allocation, where all items need not be allocated among the agents.
Thus some items may be left unallocated.
 In the case of additive valuations, Caragiannis et al.~\cite{caragiannis2019envy} showed that there always exists a partial EFX allocation where every agent receives at least half the value of her bundle in an optimal {\it Nash welfare} allocation\footnote{This is an allocation that maximizes $\Pi_{i = 1}^n v_i(X_i)$.}.
 In the case of general valuations, Chaudhury et al.~\cite{chaudhury2020little} showed how to find a partial EFX allocation and a pool of unallocated items $P$ such that no one envies the pool and the cardinality of $P$ is less than the number of agents.

Whereas fair division of divisible resources is a classical topic starting from the 1940's~\cite{Steinhaus},
fair division of indivisible items has been actively studied in recent years.
There are several studies on EF1 and EFX which are the relaxations of {\it envy-freeness}~\cite{amanatidis2020maximum, caragiannis2019unreasonable, barman2018finding, plaut2020almost, bilo2018almost, caragiannis2019envy, chaudhury2020little, chaudhury2020efx}.
Another major concept of fairness is {\it maximin share} (MMS), which was introduced by Budish~\cite{budish2011combinatorial}.
It was shown in~\cite{kurokawa2018fair} that MMS allocations do not always exist, and 
there have been several studies on approximate MMS allocations~\cite{budish2011combinatorial, bouveret2016characterizing, amanatidis2017approximation, barman2017approximation, kurokawa2018fair, ghodsi2018fair, garg2018approximating}.
In addition, studies finding {\it efficient} fair allocations have attracted attention.
{\it Pareto-optimality} is a major notion of efficiency.
Caragiannis et al.~\cite{caragiannis2019unreasonable} showed that any allocation that has maximum Nash welfare is guaranteed to be Pareto-optimal and EF1.
Unfortunately, finding an allocation with the maximum Nash social welfare is APX-hard~\cite{lee2017apx}.
There are several studies on approximation algorithms for maximizing Nash social welfare~\cite{cole2018approximating, cole2017convex, chaudhury2018fair, anari2017nash, garg2018approximating, anari2018nash}.

There are many real-world scenarios where items or resources need to be divided fairly, e.g., taxi fare division, rent division, task distribution, and so on.
Spliddit (www.spliddit.org) is a fair division website, which offers a fair solution for the division of rent, goods, and credit \cite{goldman2015spliddit}.
This website implements mechanisms for users to log in, define what is to be divided, enter their valuations, and demonstrate fair division.
Since its launch in 2014, there have been several thousands of users~\cite{caragiannis2019unreasonable}.
For more details on Spliddit, we refer to the reader to~\cite{goldman2015spliddit, plaut2020almost}.
Another fair division application is {\it Course Allocation} used at the Wharton School at the University of Pennsylvania to fairly allocate courses among students~\cite{plaut2020almost, budish2017course}.

\section{Preliminaries}
\label{sec: pre}

Let $N$ be a set of $n$ agents and $M$ be a set of $m$ items.
In this paper, we assume that items are indivisible: an item may not be split among multiple agents.
Each agent $i \in N $ has a valuation function $v_i : 2^M \rightarrow \mathbb{R}_{\ge 0}$.
We assume that (i) any valuation function $v_i$ is {\it normalized}: $v_i(\emptyset) = 0$ and (ii) it is {\it monotone}:  $S\subseteq T$ implies $v_i(S) \le v_i(T)$ for any $S,T \subseteq M$.
We say that a valuation function $v_i$ is {\it additive} if $v_i(S) = \sum_{g \in S} v_i(\{g\})$ for any $S\subseteq M$.

Throughout this paper, we assume that any valuation function is additive.
Furthermore, in our proof of Theorem~\ref{thm: main}, we consider the situation where there are only two types of valuation functions $v_\a$ and $v_\b$.
That is, for each agent $i$, either $v_i=v_{\a}$ or $v_i = v_{\b}$.
Let $N_\a$ be the set of agents whose valuation functions are $v_\a$ and $N_\b$ be the set of agents whose valuation functions are $v_\b$.
To simplify notation, we denote $v_i(g)$ instead of $v_i(\{g\})$ for $g\in M$ and use $A\setminus g, A\cup g$ instead of $A\setminus \{g\}, A\cup \{g\}$, respectively.
We also denote $S \ga T$ instead of $v_\a (S) > v_\a (T)$. 
In a similar way, we use the symbols $\la, \ge_\a, \le_\a, \gb, \lb, \ge_\b$, and $\le_\b$.

For $M' \subseteq M$, an {\it allocation} $A=(A_1,A_2,\ldots, A_n)$ on $M'$ is a partition of $M'$ into $n$ disjoint subsets, where $A_i$ is the {\it bundle} given to agent $i$. 
Under this allocation, the {\it utility} to agent $i$ is $v_i(A_i)$ (agent $i$'s value for the set of items $i$ receives.)
We say that an allocation $A=(A_1,A_2,\ldots, A_n)$ on $M'$ is {\it complete} if $M' = M$.
Otherwise, we say that an allocation $A$ is {\it partial}.

Several concepts of fairness have been considered, and one of the most well-studied notions of fairness is {\it envy-freeness}.
An allocation $A$ is {\it envy-free} if for all $i,j \in N$, 
$v_i(A_i) \ge v_i(A_j). $
We say that $i$ {\it envies} $j$ if $v_i(A_i) < v_i(A_j)$.
Unfortunately, complete envy-free allocations do not always exist when items are indivisible.
Thus the relaxations of envy-freeness have been considered.
An allocation $A$ is {\it EF1} if for all $i,j \in N$ where $i$ envies $j$, 
$\exists g \in A_j ~{\rm s.t.}~ v_i(A_i) \ge v_i(A_j\setminus g). $
A complete EF1 allocation always exists even when valuation functions are not additive and it can be obtained in polynomial time \cite{lipton2004approximately}.
However, the notion of EF1 seems to be too weak compared to envy-freeness.
Therefore, a stronger concept of fairness is desirable.
\begin{definition}
An allocation $A$ is {\it EFX} if for all $i,j \in N$, 
$$\forall g \in A_j ~{\rm s.t.}~ v_i(A_i) \ge v_i(A_j\setminus g). $$
\end{definition}
That is, in an EFX allocation, agent $i$ may envy agent $j$, but removing {\it any} item from $j$'s bundle would guarantee that $i$ does not envy $j$.
We say that $i$ {\it EFX~envies} $j$ if $\exists g \in A_j  ~{\rm s.t.}~ v_i(A_i) < v_i(A_j\setminus g)$.
We next define the standard notion of Pareto domination.
\begin{definition}(Pareto domination) For an allocation $A=(A_1,A_2,\ldots, A_n)$,  another allocation $B=(B_1,B_2,\ldots, B_n)$ \emph{ Pareto~dominates} $A$ if 
\begin{align*}
&\forall i \in N, v_i(B_i) \ge v_i(A_i), ~~and \\
&\exists j \in N, v_j(B_j) > v_j(A_j). \\
\end{align*}

\end{definition}
The existing algorithms for finding an EFX allocation with charity~\cite{chaudhury2020little} or a $\frac{1}{2}$EFX allocation\footnote{An allocation $A$ is $\frac{1}{2}$EFX if for all $i,j \in N$ and any $g \in A_j, v_i(A_i) \ge \frac{1}{2}v_i(A_j\setminus g). $}~\cite{plaut2020almost} iteratively construct a Pareto dominating (partial) EFX allocation from an existing partial EFX allocation.
We will take the same approach to prove Theorem~\ref{thm: main}.

In the following, we explain some technical ideas and notions used in the earlier papers.
In what follows in this section, we assume that every agent's valuation function is additive, but we do not need the assumption that there are only two types of valuation functions.
\paragraph{Non-degenerate instances:}  An instance $I$ is a triple $\langle N, M, \mathcal{V} \rangle $, where $N$ is a set of agents, $M$ is a set of items and $\mathcal{V}=\{v_1,\dots,v_n\}$ is a set of valuation functions.
We use an assumption on instances considered in \cite{chaudhury2020efx}.

\begin{definition} An instance $I$ is {\rm non-degenerate} if for any $i \in N$ and $S,T \subseteq M$, 
$$S\neq T \Rightarrow v_i(S) \neq v_i(T).$$
\end{definition}

Let $M=\{g_1,\dots,g_m\}$ and let $\epsilon > 0$ be a positive real number.
We perturb an instance $I$ to $I_\epsilon =\langle N, M, \mathcal{V}_\epsilon \rangle $, where for any $v_i \in \mathcal{V}$ we define $v'_i \in \mathcal{V}_\epsilon$ by $v'_i(g_j) = v_i(g_j) + \epsilon \cdot 2^j$.

\begin{lemma}(Chaudhury et al. \cite{chaudhury2020efx})
Let $\delta = \min_{i \in N} \min_{S,T: v_i(S) \neq v_i(T)} |v_i(S)-v_i(T)|$ and let $\epsilon > 0$ such that $\epsilon\cdot 2^{m+1} < \delta$. Then the following three statements hold.
\begin{itemize}
\item For any  $i\in N$ and $S,T \subseteq M$, $v_i(S) > v_i(T)$ implies $v'_i(S) > v'_i(T)$.
\item $I_\epsilon$ is a non-degenerate instance. 
\item If $X$ is an EFX allocation for $I_\epsilon$ then $X$ is also an EFX allocation for $I$.

\end{itemize}
\end{lemma}
This lemma shows that it suffices to deal with non-degenerate instances to prove the existence of an EFX allocation.
In what follows, we only deal with non-degenerate instances.

\paragraph{Minimum preferred set and Most envious agent:} 
We use the notion of most envious agent introduced in \cite{chaudhury2020little}.
Consider an allocation $X=(X_1,X_2,\ldots, X_n)$ and a set $S \subseteq M$.
For an agent $i$ such that $S >_i X_i$, we define a {\it minimum preferred set} $P_X(i,S)$ of $S$ for agent $i$ with respect to allocation $X$ as a smallest cardinality subset $S'$ of $S$  such that $S' >_i X_i$. 
Note that it holds that $P_X(i,S) >_i X_i$ and $X_i \ge_i Z$ for any subset $Z \subseteq S$ of size at most $|P_X(i,S)|-1$ by the definition of the minimum preferred set.
Define $\kappa _X(i,S)$ by
\begin{eqnarray*}
  \kappa _X(i,S)=
  \left\{
    \begin{array}{ll}
      |P_X(i,S)| & {\rm if}~S >_i X_i,\\
      +\infty & {\rm otherwise}.
    \end{array}
  \right.
\end{eqnarray*}
Let $\kappa_X(S) = \min _{i \in N}  \kappa _X(i,S)$.
We define $A_X(S)$ for a set $S$ as the set of agents with the smallest values of $\kappa _X(i,S)$, i.e., 
$$A_X(S) = \{ i\in N \mid S >_i X_i ~{\rm and}~ \kappa _X(i,S)= \kappa_X(S) \}.$$
We call $A_X(S)$ the set of {\it most envious agents}.
\paragraph{Champions:}
The notion of Champions was introduced in \cite{chaudhury2020efx}.
Let $X$ be an allocation on $M'\subseteq M$ and let $g \in M\setminus M'$ be an unallocated item.
We say that $i$ {\it champions} $j$ if $i$ is a most envious agent for $X_j\cup g$, i.e., $i \in A_X(X_j \cup g)$.
When it holds that $i \in A_X(X_i \cup g)$, we call $i$ a {\it self-champion}.
Note that $A_X(X_i \cup g) \neq \emptyset$ for any agent $i$ since we have $X_i \cup g >_i X_i$ in a non-degenerate instance.

\section{Existence of EFX for Two Additive Valuations}
\label{sec: discussion}
In this section, we prove Theorem~\ref{thm: main}.
As mentioned in Section~\ref{sec: pre}, our algorithm constructs a sequence of (partial) EFX allocations in which each allocation Pareto dominates its predecessor.
From now on we consider the case where every agent has one of the two additive  valuations $v_{\a}$ and $v_{\b}$.
We can see that if $|N_{\a}| \le 1$ or $|N_{\b}|\le 1$, then Theorem~\ref{thm: main} holds:
if $|N_{\a}| = 0$ or $|N_{\b}| = 0$, we obtain an EFX allocation in the way shown in~\cite{plaut2020almost}.
In addition, if $|N_{\a}| =1$ or $|N_{\b}|= 1$ (say $|N_{\a}| =1$), we first obtain an EFX allocation $X=(X_1,\dots, X_n)$ under the condition that all agents have identical valuations $v_{\b}$.
Then, we allocate the bundle with the highest utility in the valuation $v_{\a}$ among $X_1,\dots, X_n$ to the agent in $N_{\a}$ and allocate the rest of bundles to the other agents arbitrarily.
Obviously, this allocation becomes EFX.
Thus, we assume that $|N_{\a}| \ge 2$ and $|N_{\b}|\ge 2$.
To show our main results, we will prove the following theorem.
\begin{theorem}\label{thm: potential}
Let $X$ be an EFX allocation on $M' \subseteq M$ and let $g \in M\setminus M'$ be an unallocated item.
Then, there exists an allocation $Y$ on $S \subseteq M' \cup g$ such that
\begin{itemize}
\item $Y$ is EFX, and
\item $Y$ Pareto dominates $X$.
\end{itemize}
\end{theorem}

Assuming that Theorem~\ref{thm: potential} holds,  it is easy to show Theorem~\ref{thm: main}.
Indeed, if $Y$ Pareto dominates $X$, then we have $\sum_{i\in N} v_i(X_i) < \sum_{i\in N} v_i(Y_i)$.
In other words, we can strictly improve the value of potential function $\phi(X)= \sum_{i\in N} v_i(X_i)$ while keeping EFX.
Furthermore, since there are only a finite number of allocations, there are also only a finite number of values of potential functions.
Thus this improvement terminates in a finite number of steps.
Since we can strictly improve the value of the potential function if $M' \subsetneq M$, we  finally obtain an EFX allocation on $M$.

Now our goal is to show that Theorem~\ref{thm: potential} holds.
If there is an agent $i$ such that allocating $g$ to $i$ results in an EFX allocation, 
then we can obtain $Y$ desired in Theorem~\ref{thm: potential}.
For a few special cases, it is shown in earlier papers how to obtain a Pareto dominating EFX allocation from an existing EFX allocation.
For an allocation $X$, we define the {\it envy-graph} $E_X$, where the vertices correspond to agents and there is an edge from $i$ to $j$ if $i$ envies $j$, i.e., $X_i <_i X_j$.
The envy-graph was introduced in \cite{lipton2004approximately} and the following lemmas are known.
\begin{lemma}(Lipton et~al.~\cite{lipton2004approximately})\label{fact: cyclic}
Let $X$ be an EFX allocation such that $E_X$ has a dicycle. Then, there exists another EFX allocation $Y$ such that 
$E_Y$ is acyclic and $Y$ Pareto dominates $X$.
\end{lemma}

\begin{lemma}(Chaudhury et~al.~\cite{chaudhury2020little}) \label{ob: onesource}
Let $X$ be an EFX allocation on $M' \subseteq M$ and let $g \in M\setminus M'$ be an unallocated item.
If agent $i$ champions $j$ and $i$ is reachable from $j$ in $E_X$ (possibly $i=j$), 
then there is an EFX allocation $Y$ Pareto dominating $X$.
\end{lemma}

Lemma~\ref{fact: cyclic} implies that if the envy-graph $E_X$ has no sources, then we can obtain $Y$ desired in Theorem~\ref{thm: potential}.
Here, a {\it source} is a vertex with no incoming edges.
In addition, if there exists a self-champion, we can obtain $Y$ desired in Theorem~\ref{thm: potential} by applying Lemma~\ref{ob: onesource} as $i=j$.
We also see that if the envy-graph $E_X$ has exactly one source, then we can obtain $Y$ desired in Theorem~\ref{thm: potential} under the assumption that $E_X$ is acyclic and for any agent $i$, there exists an agent $j$ such that $j$ EFX envies $i$ when we allocate $X_i \cup g$ to $i$.
To see this, let $s$ be the unique source in $E_X$.
Then, by our assumption, there exists an agent $j\neq s$ such that $j$ EFX envies $s$ when we allocate $X_s \cup g$ to $s$.
Since agents except $s$ are not sources in $E_X$ and $E_X$ is acyclic, 
$j$ is reachable from $s$ in $E_X$, which means that we can apply Lemma~\ref{ob: onesource} to obtain $Y$.
To summarize the discussion so far, the following lemma holds.
\begin{lemma}
Let $X$ be an EFX allocation on $M' \subseteq M$ and let $g \in M\setminus M'$ be an unallocated item.
Then, there is an EFX allocation $Y$ Pareto dominating $X$ in the following cases.
\begin{enumerate}
\item[$(1)$] There is an agent $i$ such that no one EFX envies $i$ to allocate $X_i \cup g$ to $i$.
\item[$(2)$] $E_X$ has a dicycle.
\item[$(3)$] There is a self-champion.
\item[$(4)$] $E_X$ has exactly one source and $(1)$ and $(2)$ do not hold.
\end{enumerate}
\end{lemma}
Therefore hereafter, we assume that {\it the envy-graph has more than one source}, {\it there are no self-champions}, and {\it for any agent $i$, there exists an agent $j$ such that $j$ EFX envies $i$ when we allocate $X_i \cup g$ to $i$}.
 For any agent $i \in N$, we have $X_i \neq \emptyset$ by our assumption. 
Thus, by the assumption of non-degeneracy, for any $ i,j \in N_\a$ with $i \neq j$, we have $X_i \la X_j$ or $X_i \ga X_j$. Similarly, for any $i,j \in N_\b$ with $i \neq j$, we have $X_i \lb X_j$ or $X_i \gb X_j$.
We denote $N_\a = \{\a_0, \a_1, \ldots , \a_s\}$ and $N_\b = \{\b_0, \b_1, \ldots , \b_t\}$, where $X_{\a_0} \la X_{\a_1} \la \cdots \la X_{\a_s}$
and $X_{\b_0} \lb X_{\b_1} \lb \cdots \lb X_{\b_t}$.
Recall that $s,t \ge 1$.
In the following, we only change the bundles of $\a_0$ and/or $\b_0$.
Hence, to claim that a new allocation Pareto dominates $X$, we only need to  check that both $\a_0$ and $\b_0$ are not worse off and at least one of them is strictly better off than in $X$.
Furthermore, the following lemma shows that if the new allocation Pareto dominates $X$, we only need to check that $\a_1$ and $\b_1$ do not EFX envy $\a_0$ and $\b_0$, and both $\a_0$ and $\b_0$ do not EFX envy each other to claim that the new allocation is EFX.
\begin{lemma}\label{six}
Suppose that $X'$ is a new allocation Pareto dominating $X$ obtained from $X$ by changing the bundles of $\a_0$ and/or $\b_0$. 
If $\a_1$ and $\b_1$ do not EFX envy $\a_0$ and $\b_0$, and both $\a_0$ and $\b_0$ do not EFX envy each other, then $X'$ is an EFX allocation.
\end{lemma}
\begin{proof}
Since we do not change the bundles of any agents in $N\setminus \{\a_0, \b_0\}$ and $X$ is EFX, 
there is no EFX envy between them in the new allocation.
Thus, we only need to consider EFX envy toward $\a_0, \b_0$ or from $\a_0, \b_0$.
Since $\a_0$ is not worse off than in $X$ and $X$ is EFX, $\a_0$ does not EFX envy $N\setminus \{\a_0, \b_0\}$ in $X'$.
Similarly, $\b_0$ does not EFX envy $N\setminus \{\a_0, \b_0\}$ in $X'$.
Finally, if $\a_1$ does not EFX envy $\a_0$ and $\b_0$, then neither do any agents in $N_{\a}\setminus \a_0$. 
Similarly,  if $\b_1$ does not EFX envy $\a_0$ and $\b_0$, then neither do any agents in $N_{\b}\setminus \b_0$. 
\end{proof}

To sum up, if the new allocation Pareto dominates $X$, it is enough to check that there is no EFX envy for the six relationships to claim that  it is EFX.

Our discussion can be divided into two stages.
In the first stage, we strictly improve the utilities of $\a_0$ and $\b_0$ by exchanging a subset of each bundle.
Then, the new allocation will Pareto dominate the previous allocation.
However, this change may cause someone to EFX envy $\a_0$ or $\b_0$.
Thus, in the second stage, we reduce the bundles of $\a_0$ and $\b_0$ to eliminate such EFX envy.

\subsection{Improvement of the utilities of $\a_0$ and $\b_0$ }
In this subsection, we improve the utilities of $\a_0$ and $\b_0$.
For this purpose, we use the following lemmas that show the relationships between $\a_0$ and $\b_0$.
\begin{lemma}\label{ob: noenvy}
We have $X_{\a_0} \ga X_{\b_0}$ and $X_{\b_0} \gb X_{\a_0}$.
\end{lemma}
\begin{proof}
Note that $\a_1,\dots,\a_s$ and $\b_1,\dots,\b_t$ are not sources in $E_X$.
Since we assume that $E_X$ has more than one source, both $\a_{0}$ and $\b_{0}$ are sources in $E_X$.
This shows that $X_{\a_0} \ga X_{\b_0}$ and $X_{\b_0} \gb X_{\a_0}$.
\end{proof}

\begin{lemma}\label{ob: championcycle}
We have $\a_0 \in A_X(X_{\b_0} \cup g)$ and $\b_0 \in A_X(X_{\a_0} \cup g)$.
\end{lemma}
\begin{proof}
We first show $\a_0 \in A_X(X_{\b_0} \cup g)$.
Since $X_{\b_0} \lb X_{\b_0}\cup g$, we have $A_X(X_{\b_0}\cup g) \neq \emptyset$, that is, there exists an agent $i \in N_\a \cup N_\b$ such that $i \in A_X(X_{\b_0}\cup g)$.
If $i \in N_\b$, since $X_{\b_0} \le_{\b} X_i \lb P_X(i, X_{\b_0}\cup g)$, 
we have $\kappa_X(\b_0, X_{\b_0} \cup g) \le \kappa_X(i, X_{\b_0} \cup g)=\kappa_X(X_{\b_0} \cup g)$.
Thus, $\b_0 \in A_X(X_{\b_0}\cup g)$ and this contradicts the assumption that there are no self-champions.
Therefore we have $i \in N_\a$.
Then, since $X_{\a_0} \le_{\a} X_i \la P_X(i, X_{\b_0}\cup g)$, 
we have $\kappa_X(\a_0, X_{\b_0} \cup g) \le \kappa_X(i, X_{\b_0} \cup g)=\kappa_X(X_{\b_0} \cup g)$.
Thus, $\a_0 \in A_X(X_{\b_0}\cup g)$.
In a similar way, we also have $\b_0 \in A_X(X_{\a_0} \cup g)$.
\end{proof}

By Lemma~\ref{ob: championcycle}, $X_{\a_0} \la X_{\b_0}\cup g$ and $X_{\b_0} \lb X_{\a_0}\cup g$.
Hence we can define $P_X(\a_0, X_{\b_0} \cup g)$ and $P_X(\b_0, X_{\a_0} \cup g)$.
We define a new allocation $X'$ as follows: 
\begin{align*}
X'_{\a_0} &= (X_{\a_0} \cup P_X(\a_0, X_{\b_0} \cup g)) \setminus P_X(\b_0, X_{\a_0} \cup g),& \\
X'_{\b_0} &= (X_{\b_0} \cup P_X(\b_0, X_{\a_0} \cup g))\setminus P_X(\a_0, X_{\b_0} \cup g), &\\
X'_{k} &= X_{k} & ~{\rm for~all}~ k \in  N\setminus \{\a_0, \b_0\}.&
\end{align*}

\begin{figure}[htbp]
 \begin{tabular}{cc}
 \begin{minipage}[t]{0.5\hsize}
  \begin{center}
   \includegraphics[width=60mm]{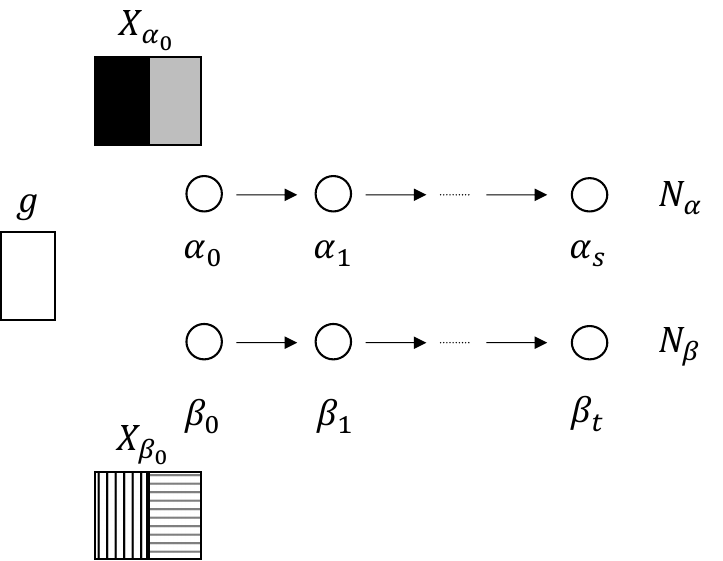}
   \captionsetup{width=.95\linewidth}
  \caption{Envy-graph $E_X$ (the edge set of $E_X$ is only partially drawn). The vertical stripe area is $P_X(\a_0, X_{\b_0} \cup g)\setminus g$ and the black area is $P_X(\b_0, X_{\a_0} \cup g)\setminus g$.}
  \label{fig: 1}
   \end{center}
 \end{minipage}
 \hfill
  \begin{minipage}[t]{0.5\hsize}
  \begin{center}
   \includegraphics[width=60mm]{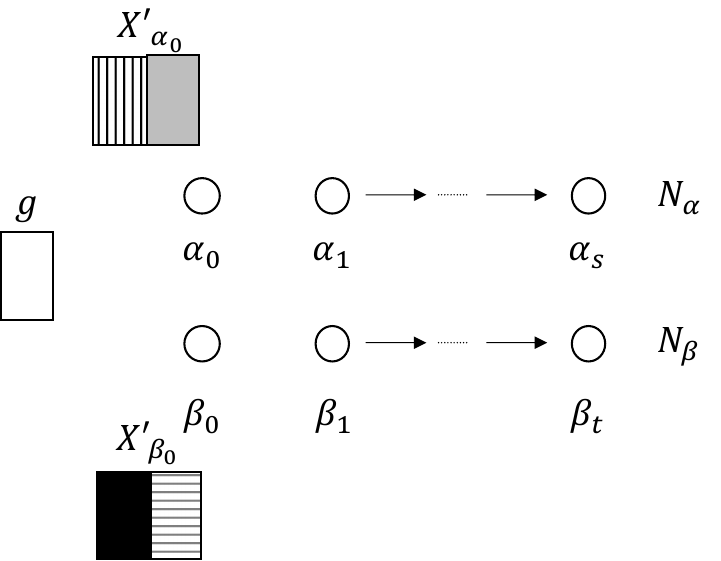}
   \captionsetup{width=.9\linewidth}
  \caption{Envy-graph $E_{X'}$ (the edge set of $E_{X'}$ is only partially drawn) obtained by exchanging $P_X(\a_0, X_{\b_0}\cup g) \setminus g$ and $ P_X(\b_0, X_{\a_0}\cup g) \setminus g$ in $X$.}
  \label{fig: 2}
  \end{center}
  
 \end{minipage}
  \end{tabular}
\end{figure}
Note that by Lemmas~\ref{ob: noenvy} and \ref{ob: championcycle}, we have $X_{\b_0} \la X_{\a_0} \la P_X(\a_0, X_{\b_0} \cup g)$ and 
$X_{\a_0} \lb X_{\b_0} \lb P_X(\b_0, X_{\a_0} \cup g)$.
This implies that $g\in P_X(\a_0, X_{\b_0} \cup g)$ and $g \in P_X(\b_0, X_{\a_0} \cup g)$. 
Thus $g$ is still an unallocated item in $X'$.
Figure~\ref{fig: 1} shows the envy-graph $E_{X}$, the bundles of $\a_{0}$ and $\b_{0}$, and an unallocated item $g$.
Figure~\ref{fig: 2} shows the envy-graph $E_{X'}$ obtained by exchanging $P_X(\a_0, X_{\b_0}\cup g) \setminus g$ and $ P_X(\b_0, X_{\a_0}\cup g) \setminus g$.
The following lemma shows that $\a_0$ and $\b_0$ are strictly better off in $X'$ than in $X$.
\begin{lemma}\label{ob: better}
We have $X'_{\a_0} \ga X_{\a_0} \ga X_{\b_0} \ga X'_{\b_0}$ and $X'_{\b_0} \gb X_{\b_0} \gb X_{\a_0} \gb X'_{\a_0}$.
\end{lemma}
\begin{proof}
By the definition of minimum preferred set, $P_X(\a_0, X_{\b_0} \cup g) \ga X_{\a_0}$.
Furthermore, since $\a_0 \not\in A_X(X_{\a_0} \cup g)$, we have $X_{\a_0} \ga P_X(\b_0, X_{\a_0} \cup g) $.
By combining these two inequalities, we obtain $P_X(\a_0, X_{\b_0} \cup g) \ga P_X(\b_0, X_{\a_0} \cup g)$.
We also see that $P_X(\a_0, X_{\b_0} \cup g) \setminus g \ga P_X(\b_0, X_{\a_0} \cup g) \setminus g$ as the valuation function is additive.
This means that an agent $\a_0$ receives the more preferred bundle from $\b_0$ and passes the less preferred bundle to $\b_0$ in her valuation. 
Thus, we have $X'_{\a_0} \ga X_{\a_0}$ and $X_{\b_0} \ga X'_{\b_0}$, and hence $X'_{\a_0} \ga X_{\a_0} \ga X_{\b_0} \ga X'_{\b_0}$ by Lemma~\ref{ob: noenvy}.
By symmetry, we also have $X'_{\b_0} \gb X_{\b_0} \gb X_{\a_0} \gb X'_{\a_0}$.
\end{proof}


By Lemma~\ref{ob: better}, $\a_0$ and $\b_0$ are strictly better off in $X'$ and the other agents do not change their bundles.
Hence, $X'$ Pareto dominates $X$. 

\subsection{Elimination of EFX envy}
In this subsection, we reduce the bundles of $\a_0$ and $\b_0$ in order to eliminate EFX envy.
As mentioned in Lemma~\ref{six}, it is enough to consider EFX envy from $\a_1, \b_0$ or $\b_1$ to $\a_0$ and from $\a_0, \a_1$ or $\b_1$ to $\b_0$.

By Lemma~\ref{ob: better}, $\a_0$ and $\b_0$ do not envy each other in $X'$.
In addition, since $X'_{\a_1} = X_{\a_1} \ga X_{\a_0} \ga X_{\b_0} \ga X'_{\b_0}$, $\a_1$ does not envy $\b_0$ in $X'$.
Similarly since $X'_{\b_1} = X_{\b_1} \gb X_{\b_0} \gb X_{\a_0} \gb X'_{\a_0}$, $\b_1$ does not envy $\a_0$ in $X'$.
Thus, there is possible EFX envy only from $\a_1$ to $\a_0$ or from $\b_1$ to $\b_0$.
If $\a_1$ envies $\a_0$, or $\b_1$ envies $\b_0$, 
then we can define $P_{X'} (\a_1, X'_{\a_0})\subseteq X'_{\a_0}$ or $P_{X'} (\b_1, X'_{\b_0}) \subseteq X'_{\b_0}$, respectively.
We define a new allocation $X''$ as follows:
\begin{align*}
X''_{\a_0} &=  \left\{
    \begin{array}{ll}
      P_{X'} (\a_1, X'_{\a_0}) & {\rm if~ \a_1~ EFX ~envies~ \a_0~ in}~ X',\\
      X'_{\a_0} & {\rm otherwise}.
    \end{array}
  \right. &\\
X''_{\b_0} &=  \left\{
    \begin{array}{ll}
      P_{X'} (\b_1, X'_{\b_0}) & {\rm if~ \b_1~ EFX ~envies~ \b_0~ in}~ X',\\
      X'_{\b_0} & {\rm otherwise}.
    \end{array}
  \right. &\\
X''_{k} &= X'_{k} & ~{\rm for~all}~ k \in  N\setminus \{\a_0, \b_0\}.&
\end{align*}
We will show that $X''$ Pareto dominates $X$, and $X''$ is EFX.
By the definition of the minimum preferred set, we have  $P_{X'} (\a_1, X'_{\a_0}) \ga  X'_{\a_1}$.
Also note that $X'_{\a_1} = X_{\a_1} \ga X_{\a_0}$.
Thus, $P_{X'} (\a_1, X'_{\a_0}) \ga X_{\a_0}$.
We also have $X'_{\a_0} \ga X_{\a_0}$ by Lemma~\ref{ob: better}.
Therefore, we have $X''_{\a_0}\ga X_{\a_0}$.
Similarly, we have $X''_{\b_0}\gb X_{\b_0}$.
Since $X''_{k} = X'_{k} = X_{k}$ for all $k \in  N\setminus \{\a_0, \b_0\}$, $X''$ Pareto dominates $X$.
We next show that $X''$ is EFX.
By Lemma~\ref{six}, it is enough to consider EFX envy from $\a_1, \b_0$ or $\b_1$ to $\a_0$ and from $\a_0, \a_1$ or $\b_1$ to $\b_0$.

\begin{itemize}
\item {\it $\a_1$ does not EFX envy $\a_0$}: 
If $X''_{\a_0}= X'_{\a_0}$, then by the definition of $X''_{\a_0}$, $\a_1$ does not EFX envy $\a_0$.
Thus we assume that $X''_{\a_0}= P_{X'} (\a_1, X'_{\a_0})$.
By the definition of the minimum preferred set, we have $P_{X'}(\a_1,X'_{\a_0})\setminus h \la X'_{\a_1} = X''_{\a_1}$ for any item $h \in P_{X'}(\a_1,X'_{\a_0})$.
Thus, $\a_1$ does not EFX envy $\a_0$ in $X''$.

\item {\it $\b_1$ does not EFX envy $\b_0$}: 
By symmetry, it is the same as above.

\item {\it $\a_1$ does not envy $\b_0$}:
By Lemma~\ref{ob: better}, we have $X''_{\a_1} =X_{\a_1} \ga X_{\a_0} \ga X'_{\b_0}$.
Since $X''_{\b_0} \subseteq X'_{\b_0}$, it holds that $X''_{\b_0} \le_{\a} X'_{\b_0}$.
By combining above two inequalities, we obtain $X''_{\a_1}\ga X''_{\b_0}$.
Thus, $\a_1$ does not envy $\b_0$.
\item {\it $\b_1$ does not envy $\a_0$}:  
By symmetry, it is the same as above.

\item {\it $\a_0$ does not envy $\b_0$}:
If $X''_{\a_0}= X'_{\a_0}$, then by Lemma~\ref{ob: better}, we have $X''_{\a_0}= X'_{\a_0} \ga X'_{\b_0} \ge_{\a} X''_{\b_0}$.
Thus, $\a_0$ does not envy $\b_0$.
If $X''_{\a_0}=P_{X'} (\a_1, X'_{\a_0})$, by the definition of the minimum preferred set, we have $P_{X'} (\a_1, X'_{\a_0}) \ga  X'_{\a_1}$.
In addition, by Lemma~\ref{ob: better}, $X'_{\a_1}=X_{\a_1} \ga X_{\a_0} \ga X'_{\b_0} \ge_{\a} X''_{\b_0}$.
By combining above two inequalities, we obtain $X''_{\a_0} \ga  X''_{\b_0}$.
Thus, $\a_0$ does not envy $\b_0$.
\item {\it $\b_0$ does not envy $\a_0$}: 
By symmetry, it is the same as above.
\end{itemize}
Therefore $X''$ Pareto dominates $X$, and $X''$ is EFX.

As a result of our discussion, it follows that Theorem~\ref{thm: potential} holds.
Therefore, Theorem~\ref{thm: main} also holds.

\section{Conclusion}
\label{sec: conclusion}
In this paper, we have shown that complete EFX allocations always exist when every agent's valuation function is additive and of two types.
This result extends the result shown in~\cite{plaut2020almost} in the sense that the number of valuation types becomes two under the assumption that each valuation is additive.
Our proof is constructive and our goal is achieved by iteratively obtaining a Pareto dominating EFX allocation from an existing EFX allocation.
We note that as shown in~\cite{chaudhury2020efx}, this method no longer works when there are three agents with additive valuations.
Therefore, a more flexible approach will be needed to deal with more general cases.
The major open problem is whether an EFX allocation always exists in the general case and 
we believe that our result is a sure step for solving this standing open problem.
The next step would be to study the case where every agent's valuation function is additive and of three given types or the case where every agent's valuation function is general and of two given types. 
\subsection*{Acknowledgments}
The author would like to thank Yusuke Kobayashi for his generous support and useful discussion.
We also thank anonymous reviewers for their helpful comments, which
improve the proof of the main theorem.
This work was partially supported by the joint project of Kyoto University and Toyota Motor
Corporation, titled ``Advanced Mathematical Science for Mobility Society''.
\bibliography{efx} 
\bibliographystyle{plain} 

\end{document}